\documentclass{llncs}

\usepackage{url}
\usepackage{graphicx}
\usepackage{amsmath}
\usepackage{amsfonts}
\usepackage{amssymb}

\usepackage{multirow}

\begin{document}

\title{Can we reach Pareto optimal outcomes using bottom-up approaches?}

\author{ Victor Sanchez-Anguix\inst{1} and Reyhan Aydo\u{g}an\inst{2,4} and Tim Baarslag\inst{3} and Catholijn M. Jonker\inst{4} }
\institute{Coventry University, Coventry, United Kingdom, \email{ac0872@coventry.ac.uk} \and
 \"{O}zye\u{g}in University, Istanbul,Turkey, \email{reyhan.aydogan@ozyegin.edu.tr} \and
 Centrum Wiskunde \& Informatica, Amsterdam, Netherlands, \email{T.Baarslag@cwi.nl} \and
 Technical University of Delft, Delft, Netherlands, \email{C.M.Jonker@tudelft.nl}}

\maketitle

\begin{abstract}

Traditionally, researchers in decision making have focused on attempting to reach Pareto Optimality using horizontal approaches, where optimality is calculated taking into account every participant at the same time. Sometimes, this may prove to be a difficult task (e.g., conflict, mistrust, no information sharing, etc.). In this paper, we explore the possibility of achieving Pareto Optimal outcomes in a group by using a bottom-up approach: discovering Pareto optimal outcomes by interacting in subgroups. We analytically show that Pareto optimal outcomes in a subgroup are also Pareto optimal in a supergroup of those agents in the case of strict , transitive, and complete preferences. Then, we empirically analyze the prospective usability and practicality of bottom-up approaches in a variety of decision making domains.
\keywords{pareto optimality, agreement technologies, group decision making, multi-agent systems, artificial intelligence}
\end{abstract}

\section{Introduction}
\label{sec:introduction}

Group decision making has been studied within different disciplines with aim of reaching a mutually acceptable outcome. One of the desired properties of that outcome is Pareto optimality. However, reaching Pareto optimal agreements is not straightforward in practice. In open and dynamic environments, decision makers may not know each other's preferences completely. It may even be the case that it becomes more complicated to find Pareto optimal solutions when the number of participants increases, as the number of interactions required to achieve an optimal deal for the group may increase due to internal conflicts or lack of trust.

A number of works in the field focus on finding a global Pareto optimal solution by involving all agents at the same time \cite{ehtamo01,sycara06,hara13,sycara2015}, which may lead to complicated interactions and lengthy decision making processes. However, we believe that, in many situations, agents can benefit from taking a bottom-up approach: calculating Pareto optimal outcomes in subgroups. In other words, we pursue the question of whether or not it is possible to estimate some Pareto optimal outcomes without knowing or predicting the preferences of all agents. In essence, solving the Pareto optimal set problem in a smaller group may be less complicated than in larger groups (e.g., less privacy concerns, less interactions needed, more willingness to cooperate, etc.) and it may provide a relatively important ratio of the final Pareto Optimal outcomes. Such kind of property can be used in some complex group decision making scenarios. Imagine that a group of agents is negotiating in unison with an unknown opponent \cite{team14,team13a,team12a}. If the agents can find the Pareto optimal outcomes within the team, they may use these outcomes in their bidding strategy to reach a Pareto optimal agreement with their opponent.

In this paper we explore bottom-up strategies. For that, first we prove that any Pareto optimal outcome in a subgroup is also Pareto optimal in a larger group that contains the subgroup, as long as agents' preferences are strict linear order. Second, we empirically simulate how bottom-up approaches may perform in realistic scenarios. More specifically, we show that we can obtain a reasonable ratio of the Pareto optimal outcomes within a group of agents by only finding the Pareto optimal outcomes within the subgroup of these agents.

The remainder of this paper is organized as follows: first we present a proof of how Pareto optimal solutions in subgroups are also Pareto optimal in larger groups when agents have strict, transitive, and complete preferences. Section~\ref{sec:applications} discusses some of the implications of the proof, and how it can be applied to solve a wide variety of problems in multi-agent systems. In Section~\ref{sec:experiments}, we empirically validate the theory in practice and analyze empirically compare the ratio of Pareto optimal outcomes within subgroups to the Pareto optimal outcomes within the entire group in a wide variety of real domains. After discussing the related work, we finally conclude the paper with future lines of work. 

\section{Pareto optimality in subgroups}
\label{sec:proof}
In this section we prove that any Pareto optimal outcome in a subgroup of agents is also Pareto optimal in any group of agents containing the subgroup. First, we provide some of the necessary definitions and introduce some notation.

Let $\mathcal{A}=\{a_1,...,a_n\}$ be a set of agents where $k$ is the index of agent $a_{k}$ and $\mathcal{A'}=\{a_1,...,a_m\}$ be a superset of $\mathcal{A}$, $\mathcal{A} \subset \mathcal{A'}$ where $m > n$. $\mathcal{O}$ is the set of all possible solutions in a given domain, and $o \in \mathcal{O}$ represents a possible solution in the domain. We assume that $\succeq_{i}$ represents agent's $a_{i}$ preference relation over outcomes in $\mathcal{O}$. If $o \succeq_{i} o'$ then agent $a_{i}$ likes $o$ at least as well as $o'$, we write  $o \succ_{i} o'$ to denote a strict preference for $o$ and $o = o'$ to denote indifference. We assume that the agents' preference relations are strict, transitive and complete. 

An outcome $o^*$ is Pareto optimal with respect to $\mathcal{A}$ and $\mathcal{O}$, denoted by $po(o^*, \mathcal{A}, \mathcal{O})$ iff \[\nexists o \in \mathcal{O}\ \exists j \leq n \overset{n}{\underset{i=1}{\bigwedge}} o \succeq_{i} o^* \wedge o \succ_{j} o^*.\] 
We denote the set of all Pareto optimal outcomes over $\mathcal{A}$ by $\mathcal{O}^*_\mathcal{A} =\left\{ o^* \in \mathcal{O} \mid po(o^*, \mathcal{A}, \mathcal{O}) \right\}.$

\begin{theorem}
Given a set of outcomes $\mathcal{O}$. For all two sets of agents $\mathcal{A}$ and $\mathcal{A'}$, if $\mathcal{A} \subset \mathcal{A'}$, then $\mathcal{O}^*_\mathcal{A} \subset \mathcal{O}^*_\mathcal{A'}$.
\end{theorem}

\begin{proof}
Let us assume by reductio ad absurdum that $\mathcal{A} \subset \mathcal{A'}$, but $\mathcal{O}^*_\mathcal{A} \not \subset \mathcal{O}^*_\mathcal{A'}$. This means there exists an $o^* \in \mathcal{O}^*_\mathcal{A}$ such that $o^* \notin \mathcal{O}^*_\mathcal{A'}$. Expanding the definition of Pareto optimal outcomes, we have
\begin{displaymath}
o^* \notin \{ o \in \mathcal{O} \mid \nexists o' \in \mathcal{O}\ \exists k \leq m, \overset{m}{\underset{i=1}{\bigwedge}} o' \succeq_{i} o \wedge o' \succ_{k} o  \} .
\end{displaymath}
This means there must exist an $o \in \mathcal{O}$ and a $k \leq m$ such that $\overset{m}{\underset{i=1}{\bigwedge}} o \succeq_{i} o^* \wedge o \succ_{k} o^*.$ We consider two scenarios: either $a_k \in \mathcal{A}$ or $a_k \notin \mathcal{A}$.

\begin{itemize}
\item If $a_k \in \mathcal{A}$ then $o$ is an outcome that dominates $o^*$ over $\mathcal{A}$, which is not possible as $o^*$ is Pareto optimal over $\mathcal{A}$.

\item Otherwise, $k > n$, so we have $\overset{n}{\underset{i=1}{\bigwedge}} o \succeq_{i} o^*$. In that case, as $o^*$ is Pareto optimal over $\mathcal{A}$, the condition is only true if all of the agents in $\mathcal{A}$ are indifferent between $o$ and $o^*$. As preferences are strict, that cannot be true either. 
\end{itemize}
Since both sides lead to a contradiction, we have proven the theorem.
\end{proof}

At this point the reader may be wondering how the theorem above behaves in a scenario where agents' preferences are not strict. As we will discuss later, the likeliness of such as scenario is small, but the conclusion of the theorem above may in fact not hold in that case. Basically, an outcome that is Pareto optimal in a subgroup $\mathcal{A}$ may not be Pareto optimal in the group $\mathcal{A'}$ when all of the agents in $\mathcal{A}$ are indifferent between such outcome and another Pareto optimal outcome. Then, one of the two outcomes may not be Pareto optimal with $\mathcal{A'}$ when one of the agents in the group is not indifferent between those outcomes. Nevertheless, as we shall outline in Section~\ref{sec:applicability}, such situations are rare in practice, as all of the agents need to be indifferent between outcomes. This becomes increasingly unlikely as the group size grows and thus, for large enough groups, we can consider that the theorem is true for practically any scenario.

\section{Prospective applications}
\label{sec:applications}
In Section \ref{sec:proof} we have demonstrated that an outcome that is Pareto optimal in a subgroup of agents will also remain Pareto optimal in a larger group\footnote{For strict, transitive, and complete preferences}. It should be highlighted that we are not depicting achieving Pareto optimality as a simple task. However, there is value in computing Pareto optimality in smaller groups as long as we are able to use those solutions in more challenging scenarios:

\begin{itemize}

 \item \textbf{Negotiation teams:} In this scenario, a group of individuals negotiate as a party with opponent(s) to achieve a deal \cite{team14,team13a,team13b,team12a,team12b}. In that case, finding the outcomes that are Pareto optimal within the team may play in favor of the team as (i) if the team sticks to these outcomes while negotiating with opponents, it can ensure efficiency in the final outcome, (ii) the set of calculated deals may be reused in multiple negotiations with different opponents as they remain Pareto optimal, and (iii) finding Pareto optimal outcomes once may reduce the time spent in negotiation threads as the team exactly knows which outcomes are more beneficial for team members. On top of that, one can also assume that team members may be more willing to share information with teammates, which may make easier the search for Pareto optimal outcomes inside the team.
 \item \textbf{Multi-party negotiations:} Some participants in a multi-party negotiation \cite{sycara2015,de2015,aydogan14,esparcia13,hara13,ehtamo01} may decide to collude and bias the agreement with their preferences. For that, the subgroup of agents may calculate Pareto optimal outcomes within the subgroup, and decide on the Pareto optimal outcomes that they plan to use in the upcoming multi-party negotiation. This way, there may be higher probabilities for the negotiation to finish with an outcome that satisfies the subgroups' interests and that is efficient. Another possible application for this proof in multi-party settings is precisely the idea of looking for Pareto optimal agreements within subgroups of agents. For instance, agents with high degrees of trust may decide to share some information that facilitates the search of Pareto optimal outcomes within the subgroup. Then, once outcomes are found in subgroups, these may be shared among all of the agents, and the whole group may need to decide on the most appropriate Pareto optimal outcome.
 \item \textbf{Decision making in open environments}: Open multi-agent systems \cite{argente11,mamsy,hewitt91} have the particularity of being systems where agents enter and leave the system dynamically. In such environments, decision making tasks may suffer from the same characteristic and agents may enter and leave decision making tasks as needed, resulting in a real time problem. For those situations, agents in a decision making task may benefit from a continuous search for Pareto optimal outcomes. As new agents join the task, those Pareto optimal outcomes calculated should be kept as they will remain Pareto optimal in the new group. When agents leave, remaining group members can get rid of some outcomes that have become dominated in the new setting.
\end{itemize}

As the reader may have noticed, the range of applications where this approach could be applied is varied. We are not claiming that those are the sole applications for this approach, and there may be others in domains like social choice, group recommendations, and so forth.

\section{Experimental study}
\label{sec:experiments}
Section \ref{sec:proof} shows theoretically that Pareto optimal outcomes within a group of agents having complete, transitive and strict preferences are still Pareto optimal when the group size increases with incoming agents. Even when preferences are non strict, we expect for the theorem to hold in most of the cases. In this section, we empirically analyze the prospective performance and applicability of bottom-up approaches. For that purpose, we selected a variety of domains:

\begin{itemize}
 \item \textbf{Sushi domain:} 5000 preference profiles over 10 types of sushi \cite{sushi03}.
 \item \textbf{AGH course selection:} 153 students' preferences over 6 courses offered by AGU University of Science and Technology in 2004 \cite{agh13}.
  \item \textbf{Book crossing domain:} The original dataset contains preferences of 278,858 users that produced 1,149,780 ratings over 271,379 books \cite{book05}. In order to calculate Pareto optimality, we require preferences to be complete on at least a subset of items. We kept 7 users that had rated 23 books in common.
 \item \textbf{Movielens domain:} The original dataset contains 138,000 users that provided ratings over 27,000 movies \cite{movielens03}. As we require complete preferences, we picked 10 preference profiles that had rated a total of 298 movies in common.
 \item \textbf{Holiday domain:} A multi-party negotiation domain available in Genius \cite{genius09}. In this scenario, participants need to decide on the details of a holiday trip. In total, 9 preference profiles over 1024 possible outcomes are available. These preferences have been elicited from TU Delft computer science students, but not with serious plans for a joint holiday in mind.
 \item \textbf{Symposium domain:} Another multi-party negotiation domain that is available in Genius \cite{genius09} concerning the organization of a conference. There are 9 preference profiles over 2305 possible outcomes. These preferences have been elicited from faculty members in computer science of TU Delft experienced in organizing conferences, but not having a specific conference in mind.
  \item \textbf{Party domain:} Another multi-party negotiation domain, where agents decide on the details of a party \cite{genius09}. We elicited preferences from students in a Master level AI course. Students were asked to input their real preferences via Genius based on their tastes for hosting parties. In total, we elicited 24 real preference profiles over 3072 outcomes.
\end{itemize}

From a global perspective, the sushi, agh, and book domain are \textit{small} attending to the number of outcomes. These domains correspond to decision making domains where outcomes are non customizable objects (e.g., a movie, a book, a course, etc.). The data in the Movielens domain is less sparse and we were able to find 10 users that had rated 298 outcomes in common. This is again a domain where outcomes are non customizable, but the size of the domain is one order of magnitude larger than that of the small domains. The three remaining multi-party negotiation domains (i.e., holiday, symposium, and party domain) represent scenarios where the final outcome can be customized via the negotiable issues. As a result, the number of possible outcomes is larger. We consider these domains and the Movielens domain as the \textit{large domains} in our study.


\subsection{Validation and Performance Analysis}

Our performance metric is the ratio of the Pareto optimal outcomes within a subgroup with a size of \{2, ..., n-1\} to the Pareto optimal outcomes within the n-sized group.  If the ratio remains low even for large subgroups, then this means that the performance of our theoretical finding may be of little value in practice, as only a small ratio of the final Pareto outcomes may be achievable. However, if the ratio is large, then it may indicate that bottom-up approaches may be valuable. Additionally, common sense indicates that, the larger the subgroup, the higher the ratio of final Pareto optimal outcomes that may be obtained in the subgroup. However, one question that arises is the actual speed by which the ratio of final Pareto outcomes increases, and whether or not subgroups may be able to calculate a respectable ratio of the final Pareto optimal outcomes.

For testing the practical performance of our bottom-up approaches, we randomly generated groups of size $n$ based on the preference profiles available for each domain.  For each randomly generated group, we built all possible subgroups with varying sizes $\{2,..,n-1\}$ and estimated the Pareto optimal set in each (sub)group. More specifically, for each domain we tested a maximum of 1000 groups\footnote{The total number is $min(1000,\binom{m}{n})$, where $m$ is the total number of available preference profiles and $n$ is the size of the group} of size $n= \{5,7,9\}$.

\begin{figure*}[ht]
\includegraphics[width=\linewidth]{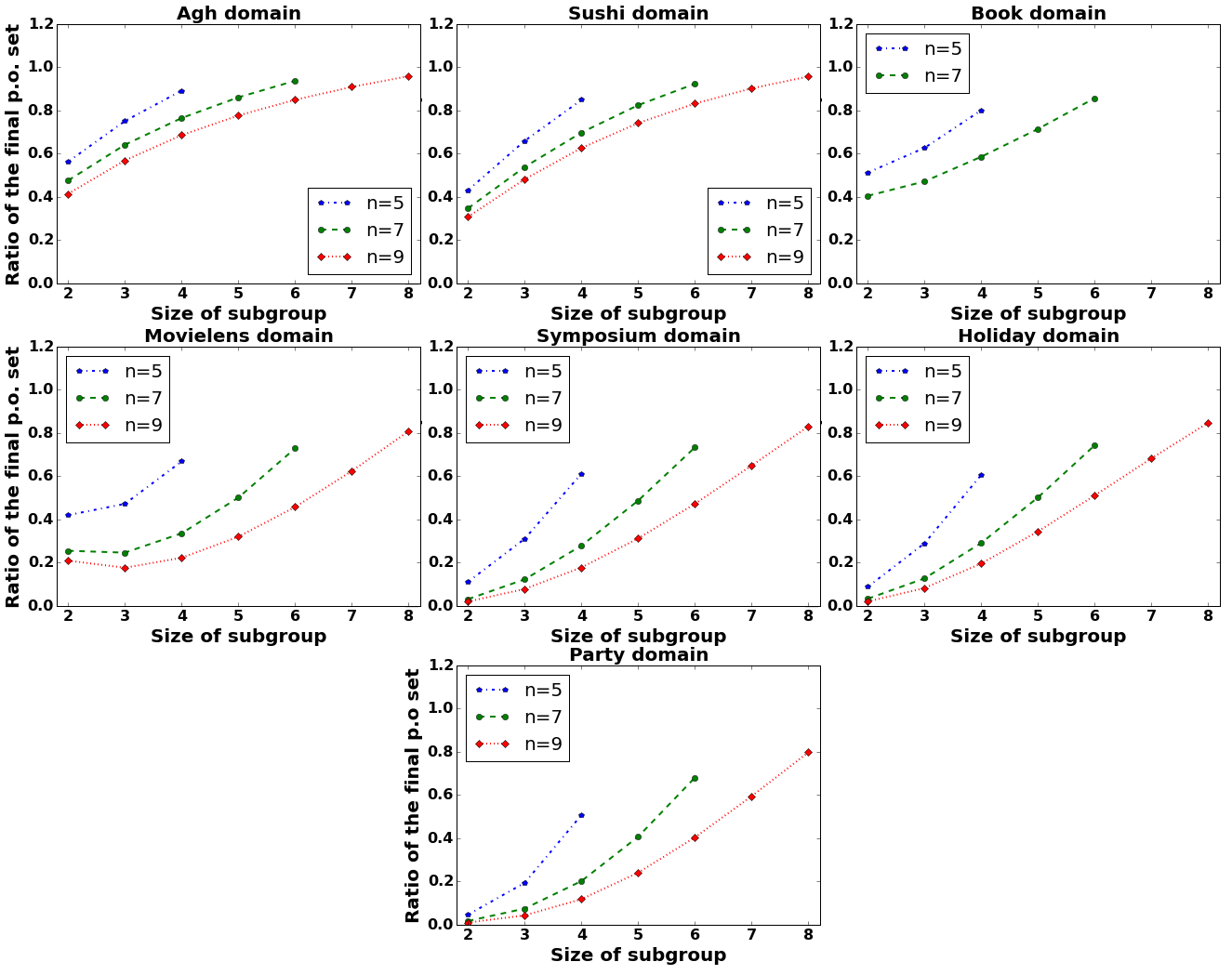}
 \caption{Average ratio of the final Pareto optimal obtained in subgroups of different size}
 \label{fig:evolution}
\end{figure*}

The results of this experiment can be observed in Figure \ref{fig:evolution}. As expected, the results show that the larger the subgroup is, the larger the average ratio of the final Pareto Optimal set that we get. The increase is clearly continuous for all of the domains and group sizes. When we look at the results for groups of 7 and 9 members we observe a non-linear increase with the size of the subgroup. This non-linear increase is not as evident in the case of 5 members' groups, as in that case we only have 3 data points\footnote{Even non-linear functions may look like linear when the number of points is reduced}. 

One should highlight that for $n-1$ agents in the subgroup, $n$ being the total number of agents in the group, the average ratio of the Pareto optimal set obtained in the subgroup is always over 50\% of the final set, being close to 80\% in some cases (e.g, smaller domains, larger groups). This is a good result, especially for negotiation team scenarios \cite{team14,team13a,team12a}, where the team could calculate the Pareto set inside the team and use those outcomes in the negotiation with an opponent. This is a clear case where a subgroup of size $n-1$ can be formed (i.e., all of the team members) and, according to the experimental results, obtain a notably high ratio of final Pareto optimal outcomes. Consequently, they can propose Pareto optimal bids without knowing their opponent's preferences.

The result is also notable for smaller subgroups. For instance, in groups of size 5, we are able to obtain between an average of 68\% of the final Pareto set for small domains and 32\% for the larger domains with just about half of the group members (i.e., 3). In the case of groups of size 7, we get 68\% of the final Pareto set for small domains and 28\% for larger domains with just about half of the group members (i.e., 4). Similarly, for groups of size 9 we are able to obtain an average of 76\% of the final Pareto set in small domains, and 30\% in large domains  with just about half of group members (i.e., 5).


The trends in the graphics and the results mentioned above may also suggest that larger domains may result in lower ratios of the final Pareto optimal set achievable by subgroups. Nevertheless, as we have been able to observe above, the results can still be considered as positive. Although the current results are promising, we would like to test a wider range of domains and domain sizes to strengthen the results of this study.

\subsection{Applicability Analysis}
\label{sec:applicability}
There are still other aspects that we need to analyze to determine the applicability of bottom-up approaches in real situations. Even though considerable ratios of the final Pareto optimal set are obtainable within subgroups, this may be useless in practice if the total number of Pareto optimal outcomes is very close to all possible outcomes. In those cases, there would be no point in calculating Pareto optimal outcomes in subgroups, as almost any outcome would be Pareto optimal. Therefore, we are interested in checking that the set of final Pareto optimal outcomes does not dramatically approach the total number of outcomes.  In \cite{barry81}, O'Neill studied how Pareto optimality was affected by the number of agents participating in a decision making process. To put it simply, the author proved that the number of Pareto optimal outcomes grows exponentially with the number of agents, with the assumption that all preference profiles are equally probable. Additionally, he proposed a formula to estimate the number of outcomes that are expected to be Pareto optimal based on the size of the domain $m$, and the number of agents in the group $n$: $E(K_{m,n}) = - \overset{m}{\underset{ i=1 }{\sum}} (-1)^i \binom{m}{i} \frac{1}{i^{n-1}}$.

He also stated that the size of the domain had an effect on the number of outcomes that were Pareto optimal: larger outcome spaces tend to slow down the exponential growth of the Pareto optimal set, although the growth is still exponential. Of course, for drawing such a conclusion, the author had to assume that all preference profiles were equally probable. We argue that, in practice, all preference profiles are not equally probable as in some domains not all of the outcomes may be equally feasible (e.g., high prices in a team of buyers, popular choices in movies, popular choices in travel destinations, etc.). Hence, we argue that the exponential growth may not be as fast as in the theoretical case, and bottom-up approaches may be applicable to more scenarios.

In order to examine this theoretical finding in practice, we calculated the ratio of the Pareto optimal outcomes to the total number of outcomes for each domain and group size. Figure \ref{fig:exponential} shows the average ratio of outcomes that are Pareto optimal for different groups sizes and domains. In these graphs, blue dots represent the average ratios calculated in real scenarios while green dots denote the theoretical estimation provided by \cite{barry81} for domains of the same size. In addition to this, for each data point we provide the total number of cases\footnote{Again, the total number is $min(1000,\binom{m}{n})$}  that were considered for calculating the average. Numbers in red represent less than 30 samples and such averages should be ignored.

\begin{figure}[ht]
\includegraphics[width=\linewidth]{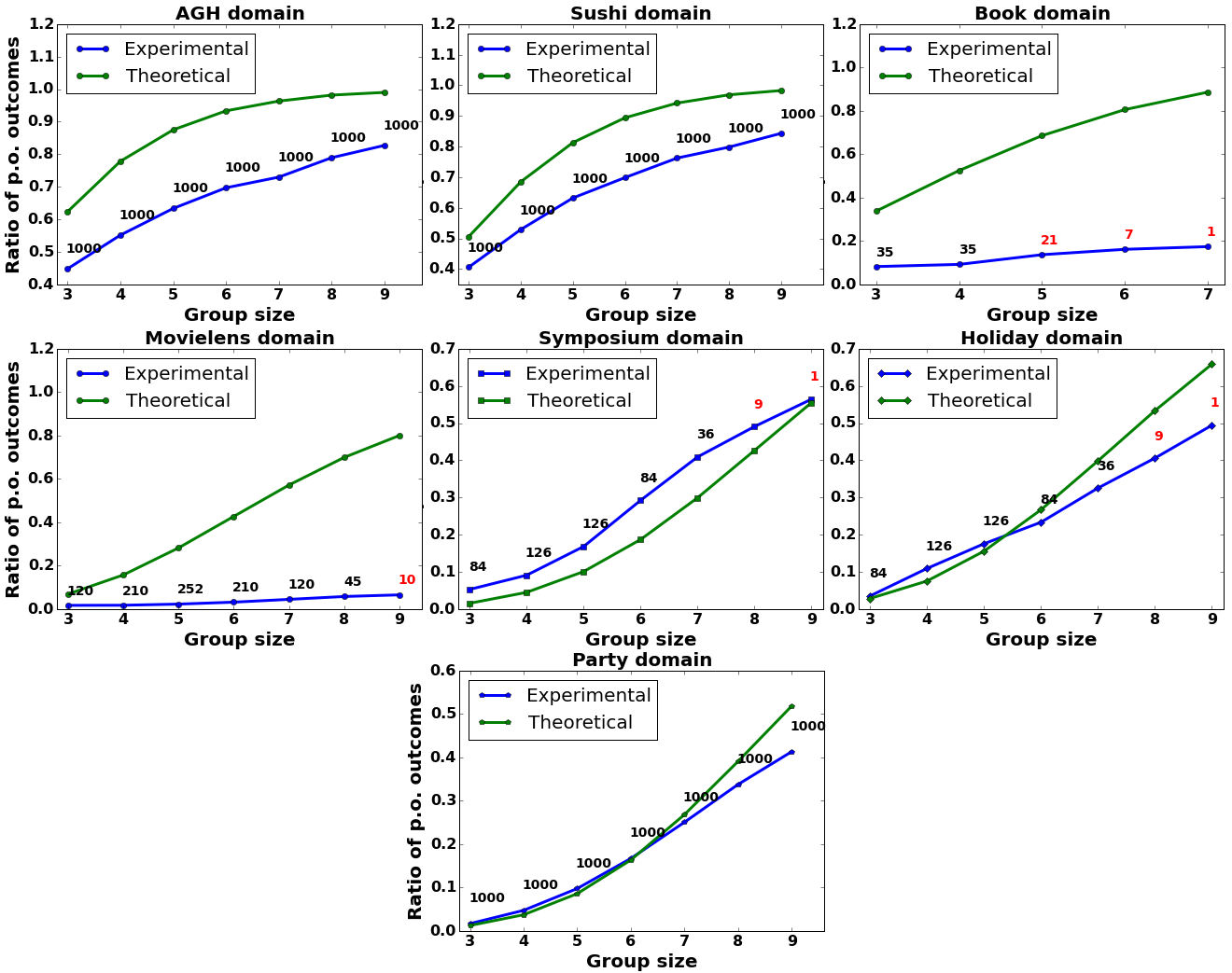}

 \caption{Average ratio of the final Pareto optimal set obtained with subgroups of different size}
 \label{fig:exponential}
 \vspace{-0.5cm}
\end{figure}

As it can be observed in Figure \ref{fig:exponential}, the growth in the number of outcomes that are Pareto optimal is usually slower in real domains than in the theoretical estimation. Being more specific, we observe that only the symposium domain shows a similar behavior to that of the theoretical case. The rest of the domains deviate from the theoretical behavior sooner or later, showing a slower saturation. We can observe that this difference is specially acute in the Movielens, Book, Sushi, and Agh domain, which are the ones whose preferences have been rigorously elicited from real users (except for the party domain).  This may reinforce our initial intuition, that, in real domains, the exponential growth on the number of Pareto optimal outcomes may not be as drastic as in the theoretical case. In other domains like the party and the holiday domain, the difference is less acute but still existent.

In fact, if one analyzes the proposed domains one by one, it is possible to realize that there are general preferential trends. This is clear in domains like Movielens or the Book domain, where we know that some movies and some books tend to be more popular than others. For instance, \textit{The Shawshank Redemption} is one of the most popular movies of all times, and it has been able to obtain average ratings of 9.3 over 10 stars in sites like \textit{IMDB}.\footnote{\url{http://www.imdb.com}. Visited on 16th November 2015}, where it has been voted by more than 1 million users. Similarly, we can find books like \textit{Harry Potter and the Deathly Hallows} that have received an average rating of 4.59 over 5 stars with more than 1 million ratings on sites like \textit{GoodReads}\footnote{\url{http://www.goodreads.com/}. Visited on 16th November 2015}. Finding users that did not like these items has low odds, and as a consequence we can state that not all preference profiles are equally probable. Not only there are general trends in users preferences, but many times we find that there are clusters of users with similar preferences. For instance, in the book domain, we can expect that users that have liked \textit{The Lord of the Rings} will also like other fantasy themed books like \textit{Song of Ice and Fire}. This is the type of patterns exploited by recommender systems, and suggests that the number of likely preference profiles is even smaller.

With respect to the other small domains (e.g., AGH, Sushi), we analyzed the preferences of users. In fact, for analyzing the preferences of users on items we performed a Borda count with all of the preference profiles. We could observe that, in the Sushi domain, there are also some popular choices the \textit{toro} (a total score of 39445)  and some choices that are usually the least liked by users like the \textit{kappa-maki} (a total score of 14928). In the case of the AGH domain, we could also observe that one of the courses (e.g., course 3) was the most preferred one with a score of 731, whereas the least preferred score had almost half the score. This means that in these domains, preferences are not equally distributed and one should not expect such an exponential growth as in the theoretical case.

 With respect to negotiation domains, we elicited real preferences from the Party domain, whereas we used the preference profiles provided by Genius in the Holiday and Symposium domain. Interestingly, we could observe that real users in the Party domain tend to consider the type of food, the type of drinks, and the music as the most important attributes. Even in some specific attributes, we could find that there were popular choices like for instance \textit{Beer only} for drinks, and \textit{Finger-food} and \textit{Chips and Nuts} for food choices. With respect to the rest of negotiation domains, it has to be considered that they were not strictly and rigorously elicited like in the case of the party domain. Users were not contextualized in a specific scenario and their preferences were just elicited from their previous experiences in similar scenarios. In the case of the Holiday domain, we were able to observe some patterns like users considering the duration and the activities as the most important attributes. The users usually preferred longer durations to shorter durations, and we observed a slight positive inclination towards \textit{Historical Places} and \textit{Restaurants}. Even in the rest of less important attributes we were able to find some patterns like the fact that most users preferred \textit{Miami} and \textit{Amsterdam} as destinations. These patterns again show that not all preference profiles are equally likely, and that is reflected in the fact that the experimental growth depicted for Figure \ref{fig:exponential} is slower than the theoretical growth.  On the other hand, the Symposium domain is closer to the theoretical expectation.  This may be explainable due to the fact that the Symposium domain preferences were not elicited thinking on an specific symposium. In contrast to the Holiday domain, which did not follow a rigorous preference elicitation process either, in the Symposium domain it is harder to relate to the scenario, as it includes totally fictional speakers (e.g. \textit{Mr. Talkolot}), whereas in the holiday domain one always can think about his/her own preferences on a trip. This may explain why the increase in the ratio of Pareto optimal outcomes is similar to the theoretical case where preference profiles are equally probable. It should be highlighted that in many negotiation domains, preferences are made different to test the performance of negotiation algorithms in conflicting scenarios.

The fact that, as we have shown, not all preference profiles are equally likely makes bottom-up approaches more applicable to real life scenarios than the results depicted in theory\cite{barry81}. However, it should be noted that, even though the growth is slower, the graphics still suggest an increase with the size of the group and eventually the proof may not be applicable for domains involving a large number of agents. These results raise an interesting trade-off that should be analyzed in the future: the relation between the performance of bottom-up approaches, which increases with the subgroup size, and its applicability, which decreases with the group size, as nearly all outcomes may be Pareto optimal.

\begin{table}[]
\centering
\begin{tabular}{|l|l|l|l|l|l|l|l|}
			    \hline
 Group size  & \multicolumn{7}{|c|}{Subgroup size} \\ \hline
                      & 2   & 3  & 4  & 5  & 6  & 7  & 8  \\ \cline{2-8}
 5                    & 7\%    & 4\%   & 1\%   &  -  &  -  &  -  &  -  \\ \cline{2-8}
 7                    & 5\% &  2\%  & 1\%   & 0.7\%   & 0.3\%   &  -  &  -  \\ \cline{2-8}
 9                    & 4\%    & 2\%   & 1\%   & 0.7\%   & 0.4\% & 0.2\%  &   0.07\% \\ \hline
\end{tabular}
\caption{Average \% of false positives calculated in a subgroup}
\label{tab:false}
\vspace{-0.8cm}
\end{table}

There is another additional issue to be studied concerning the applicability of bottom-up approaches. As the reader may have guessed, the aforementioned domains do not guarantee strict preferences. Therefore, some Pareto optimal outcomes calculated in subgroups may not be Pareto optimal in the whole group (we call these false positives). In order to study this, we measured the ratio of false positives in the previous experimental setting. The results are summarized in Table \ref{tab:false}. As it can be observed, the percentage of false positives remains low for every possible scenario, and it tends to decrease with the size of the subgroup. This matches our initial intuition, and shows that the proof presented in this paper practically holds in every situation. Hence, this result supports the applicability of bottom-up approaches in practice.

\section{Related work}
\label{sec:related}

As far as these authors are concerned, most of the studies have dedicated their efforts on reaching Pareto optimal solutions using a horizontal approach that involves interactions with all group members.  In \cite{sycara06}, the authors propose a general framework for bilateral negotiations where agents are able to reach near Pareto optimal outcomes by decomposing the negotiation process into iso-utility curves, from where outcomes are proposed based on the similarity to the last offer proposed by the opponent. Later, the authors extend their findings to a multilateral and multi-issue environment where convergence is guaranteed \cite{sycara2015}. Ehtamo \textit{et al.} \cite{ehtamo01} propose a centralized mechanism for achieving Pareto optimal outcomes based on real valued linear additive utility functions and information sharing. Amador \textit{et al.} \cite{amador14} propose a task allocation method for agents with temporal constraints that is capable of providing envy free and Pareto optimal solutions under specific conditions. Other works like \cite{rahwan08} have extended the concept of Pareto optimality to argumentation frameworks. The authors study different agent attitudes, how they relate to the problem of efficiency in abstract argumentation dialogues, and define several situations and scenarios that lead to Pareto optimal arguments. Recently, Hara \textit{et al.} \cite{hara13} proposed a mediated mechanism based on genetic algorithms that is capable of achieving near Pareto optimal outcomes for multi-party negotiations where agents preferences present non-linear relationships and change over time. However, none of these works employ bottom-up approaches, which may prove more useful in some scenarios.

Another field related to our study is that of multi-objective optimization. Pareto optimality is a well-known efficiency measure in multi-objective optimization \cite{horn94,li09,hu13}. Similarly to our multiagent decision setting, researchers in centralized multi-objective optimization have noticed the exponential increase on the number of Pareto optimal outcomes with the number of objective functions \cite{di06,corne07}. Due to this unfortunate property of Pareto optimality, some researchers have offered practical alternatives to the selection of Pareto optimal outcomes. Di Pierro \textit{et al.} define the concept of $k$ optimality for deciding over Pareto optimal outcomes. Basically, a non-dominated outcome is defined as $k$-optimal when that outcome is non-dominated over every possible combination of $k$ objectives. Thus, it results in a stronger concept of optimality that may help to choose a solution over a set of Pareto optimal outcomes. We want to highlight the practical usability of $k$-optimality on future decision making mechanisms for agents and how it complements our current findings. First, based on our proof, a subgroups of agents may calculate Pareto optimal outcomes on subgroups and communicate them to the rest of subgroups. Then, a mechanism may be devised to allow agents to select a $k$-optimal outcome over calculated Pareto optimal outcomes.

Finally, economic and theoretical studies are also a source of related work. As introduced in the text, \cite{barry81} analyzed how the number of Pareto optimal outcomes exponentially increases with the number of agents by assuming that all preference profiles are equally probable. In our present study, we have, among other contributions, shown how real domains in practice behave with regards to Pareto optimality. More specifically, we have shown that, despite the increase in the number of Pareto optimal outcomes with the number of agents, the growth speed is not as quick as portrayed by \cite{barry81}. This is, as far as we know, our closest work in the study of the underlying properties of Pareto optimality. Of course, there have been other successful studies on Pareto optimality for specific domains and problems like characterizing fairness, or studying the relationship between monotonic solutions and Pareto optimality \cite{bogo15,garcia2015}, but their focus of study has not been on the exploration of bottom-up approaches for reaching Pareto optimality.

\section{Conclusion}
\label{sec:discusion}
\vspace{-0.2cm}

In this paper, we have explored the applicability and performance of bottom-up approaches for reaching Pareto optimal outcomes in groups. Our analysis shows that Pareto optimal outcomes in a group remain optimal when increasing the number of agents in the group in many practical scenarios. This has implications for bottom-up approaches, as Pareto optimal outcomes may be calculated in subgroups first, and then be used in scenarios involving the whole group. 

We performed experimental analysis on preferences elicited from users in real-life scenarios and validated that this principle can be applied to a wide range of domains. Our results on performance and applicability indicate that we are able to calculate a significant ratio of the final Pareto optimal frontier within subgroups. Conversely, we analyzed the applicability of our approach by considering how the ratio of Pareto outcomes increases with the size of the group. Our findings highlight that this increase is not as abrupt as expected in theoretical studies, as not all preference profiles are equally likely in many real-life domains. Still, the increase of the ratio of final Pareto optimal outcomes points to a necessary trade off in practice, which we plan to analyze in the future.


Additionally, as a future work, we plan to design novel negotiation approaches for intra-team negotiations that benefit from our findings. In particular, we plan to design a negotiation strategy for negotiation teams, which first calculate the Pareto optimal solutions within the team using our approach, and then target that set of Pareto optimal proposals when negotiating with the opponent.


\bibliographystyle{abbrv}
\bibliography{aamas2016}

\end{document}